\documentclass[11pt,a4paper]{article}

\usepackage[utf8]{inputenc} % set input encoding (not needed with XeLaTeX)

\usepackage{amsmath, amsthm, amssymb}

\usepackage{url}
\usepackage{comment}
\usepackage{microtype}
\usepackage{enumitem}
\usepackage{xspace}
\usepackage{thmtools}
\usepackage[linesnumbered,ruled,noend]{algorithm2e}
\usepackage{standalone}
\usepackage{tikz}
\usepackage{tikz-cd}
\usepackage{todonotes}
\usepackage{hyperref}
\usepackage[OT4]{fontenc}
\usepackage{authblk}
\usepackage{fullpage}

\usepackage[capitalize]{cleveref}

\usepackage{color}
\definecolor{darkgreen}{rgb}{0,0.5,0}
\definecolor{darkblue}{rgb}{0,0,0.8}
\definecolor{darkred}{rgb}{0.8,0,0}

 \hypersetup{
   unicode=false,          % non-Latin characters in Acrobat’s bookmarks
   colorlinks=true,        % false: boxed links; true: colored links
   linkcolor=darkred,          % color of internal links (change box color with linkbordercolor)
   citecolor=darkblue,        % color of links to bibliography
   filecolor=magenta,      % color of file links
   urlcolor=blue           % color of external links
 }

\setlist[enumerate]{nosep}
\setlist[itemize]{nosep}

\newcommand{\Oh}{\mathcal{O}}

\newtheorem{thm}{Theorem}[section]
\newtheorem{theorem}[thm]{Theorem}
\newtheorem{fact}[theorem]{Fact}
\newtheorem{proposition}[thm]{Proposition}
\newtheorem{conjecture}[thm]{Conjecture}
\newtheorem{corollary}[thm]{Corollary}
\newtheorem{lemma}[thm]{Lemma}
\newtheorem{observation}[thm]{Observation}
\newtheorem{definition}[thm]{Definition}

\newcommand{\Ohtilde}{\tilde{\mathcal{O}}}
\newcommand{\HD}{\operatorname{HD}}

\newcommand{\Occ}{\operatorname{Occ}}

\newcommand{\sub}{\subseteq}

\newcommand{\floor}[1]{\left\lfloor #1 \right\rfloor}
\newcommand{\ceil}[1]{\left\lceil #1 \right\rceil}

\newcommand{\dd}{\mathinner{.\,.}}

\newcommand{\bigO}[1]{\mathcal{O}(#1)}
\newcommand\tO[1]{\widetilde{\mathcal{O}}(#1)}
\newcommand{\Z}{\mathbb{Z}}

\newcommand\size[1]{\lvert#1\rvert}
\newcommand\Query{\mathsf{Query}}
\newcommand\Update{\mathsf{Update}}

\crefname{corollary}{Corollary}{Corollaries}
\crefname{fact}{Fact}{Facts}

\title{The Dynamic $k$-Mismatch Problem} %TODO Please add

\author[1]{Rapha\"el Clifford}
\author[2]{Paweł Gawrychowski}
\author[3]{Tomasz Kociumaka\thanks{Partly supported by NSF 1652303, 1909046, and HDR TRIPODS 1934846 grants, and an Alfred P. Sloan Fellowship.}}
\author[4]{Daniel P. Martin}
\author[2]{Przemys\l{}aw~Uzna\'nski\thanks{Supported by Polish National Science Centre grant 2019/33/B/ST6/00298.}}
\affil[1]{Department of Computer Science, University of Bristol, UK}
\affil[ ]{\texttt{raphael.clifford@bristol.ac.uk}}
\affil[2]{Institute of Computer Science, University~of~Wroc\l{}aw,~Poland}
\affil[ ]{\texttt{\{gawry,puznanski\}@cs.uni.wroc.pl}}
\affil[3]{University of California, Berkeley, U.S.}
\affil[ ]{\texttt{kociumaka@berkeley.edu}}
\affil[4]{The Alan Turing Institute, British Library, London, UK}
\affil[ ]{\texttt{dmartin@turing.ac.uk}}

\date{}

\begin{document}
\maketitle

\begin{abstract}
The text-to-pattern Hamming distances problem asks to compute the Hamming distances between a given pattern of length $m$ and all length-$m$ substrings of a given text of length $n\ge m$.
We focus on the well-studied $k$-mismatch version of the problem, where a distance needs to be returned only if it does not exceed a threshold $k$. Moreover, we assume $n\le 2m$ (in general, one can partition the text into overlapping blocks).
In this work, we develop data structures for the dynamic version of the $k$-mismatch problem supporting two operations:
An update performs a single-letter substitution in the pattern or the text, whereas a query, given an index $i$, returns the Hamming distance between the pattern and the text substring starting at position $i$, or reports that the distance exceeds $k$. 

First, we describe a simple data structure  with $\tO{1}$ update time and $\tO{k}$ query time. 
Through  considerably more sophisticated techniques, we show that $\tO{k}$ update time and $\tO{1}$ query time is also achievable. These two solutions likely provide an essentially optimal trade-off for the dynamic $k$-mismatch problem with $m^{\Omega(1)} \le k \le \sqrt{m}$:
we prove that, in that case, conditioned on the 3SUM conjecture, one cannot simultaneously achieve $k^{1-\Omega(1)}$ time for all operations (updates and queries) after $n^{\Oh(1)}$-time initialization.
For $k\ge \sqrt{m}$, the same lower bound excludes achieving $m^{1/2-\Omega(1)}$ time per operation.
This is known to be essentially tight for constant-sized alphabets: already Clifford et al. (STACS 2018) achieved $\Ohtilde(\sqrt{m})$ time per operation in that case,
but their solution for large alphabets costs $\Ohtilde(m^{3/4})$ time per operation.
We improve and extend the latter result by developing a trade-off algorithm that, given a parameter $1\le x\le k$,
achieves update time $\Ohtilde\big(\frac{m}{k} +\sqrt{\frac{mk}{x}}\big)$ and query time $\tO{x}$.
In particular, for $k\ge \sqrt{m}$, 
an appropriate choice of $x$ yields $\Ohtilde(\sqrt[3]{mk})$ time per operation,
which is $\Ohtilde(m^{2/3})$ when only the trivial threshold $k=m$ is provided.
\end{abstract}

\newpage

%!TEX root = dynamic-kmismatch.tex
\section{Introduction}

%For $k \in \Oh(\sqrt{n})$ we show that  $U(k,n) + Q(k,n)$ cannot be $\Oh(k^{1-\varepsilon})$ unless both the 3SUM and online matrix vector conjectures are false.    We show a data structure supporting queries and updates in $\tilde{\Oh}(n^{1/3}k^{1/3})$ time, $(\tilde{\Oh}(1),  \tilde{\Oh}(k)$, $\tilde{\Oh}{k^2}, \tilde{\Oh}(1)$.

The development of dynamic data structures for string problems has become a topic of renewed interest in recent years (see, for example,~\cite{Amir2020,Amir2021,Amir2019,amir2019longest,Charalampopoulos2020,charalampopoulos2020faster,clifford2018upper,GKKLS,KK22} and references therein).
Our focus will be on approximate pattern matching, where the general problem is as follows:
Given a pattern  of length $m$ and a longer text of length $n$,  return the value of a distance function between the pattern and substrings of the text.  

We develop a new dynamic data structures for a thresholded version of the Hamming distance function, known as the $k$-mismatch function.  In this setting, we only need to report the Hamming distance if does not exceed $k$. The $k$-mismatch problem is well studied in the offline setting, where all alignments of the pattern with the text substring that meet this threshold must be found. 
In 1980s, an $\Oh(nk)$-time algorithm was given~\cite{TCS:LanVis86}, and this stood as the record for a over a decade.
However, in the last twenty years, significant progress has been made.  
In a breakthrough result, Amir et al.~\cite{SODA:AmiLewPor00} gave $\Oh(n\sqrt{k \log k})$-time and $\Oh\big(n + \frac{k^3 \log k}{m}\big)$-time algorithms, which were subsequently improved to $\Oh(n\log^{\Oh(1)}m + \frac{nk^2\log{k}}{m})$ time~\cite{SODA:CFPST16}, $\Oh\big(n\log^2{m} \log{\sigma}+\frac{nk\sqrt{\log{m}}}{\sqrt{m}}\big)$ time~\cite{gawrychowski2018towards}, and finally to $\Oh\big(n+\min\big(\frac{nk\sqrt{\log m}}{\sqrt{m}}, \frac{nk^2}{m}\big)\big)$ time~\cite{chan2020approximating}. 

 %However, for the thresholded version of Hamming distance, it turns out to be considerably more challenging to design near optimal dynamic data structures.

In the dynamic $k$-mismatch problem, there
are two input strings: a pattern $P$ of length $m$ and a text $T$ of length $n\ge m$. 
For a query at index $i$, the data structure
must return the Hamming distance between $P$ and $T[i\ldots i+m)$ if the Hamming distance is less than $k$, and $\infty$ otherwise.  
The queries can be interspersed with updates of the
form $\Update(S,i,x)$, which assign $S[i]:=x$, where $S$ can be either the pattern or the text. There are two naive approaches
for solving the dynamic problem. The first is to rerun a static offline
algorithm after each update, and then have constant-time queries. The second is to simply modify the input at each update  and compute the Hamming distance naively for each query. Our goal is to perform better than these naive solutions.

We primarily focus on the case when $n=\Theta(m)$ (in general, one can partition the text into $\Theta(\frac{n}{m})$ overlapping blocks of length $\Theta(m)$).
When $k=m$ and $\sigma=n^{o(1)}$, known upper bounds and conditional lower bounds match up to a subpolynomial factor:
There exists a dynamic data structure with an $\Oh(\sqrt{n \log n} \cdot \sigma)$ upper bound for both updates and queries and an almost matching $n^{1/2-\Omega(1)}$ lower bound~\cite{clifford2018upper}
conditioned on the hardness of the online matrix-vector multiplication problem. 
Although there is no existing work directly on the dynamic $k$-mismatch problem we consider,  it was shown very recently that a compact representation of all $k$-mismatch occurrences can be reported in $\tO{k^2}$ time\footnote{The $\Ohtilde(\cdot)$ notation suppresses $\log^{\Oh(1)}n$ factors.} after each $\Oh(\log n)$-time update~\cite{charalampopoulos2020faster}.

We give three data structures for the dynamic $k$-mismatch problem. The first has update
time of $\tO{1}$ and a query time of $\tO{k}$.  The main tool we use is the dynamic strings data structure~\cite{GKKLS} which allows enumerating mismatches in $\Oh(\log n)$ time each.  The second has update time $\tO{k}$ and a query time of $\tO{1}$.
Here, we build on the newly developed generic solution for the static $k$-mismatch problem from~\cite{charalampopoulos2020faster}. The third data structure, optimized for $k\ge \sqrt{n}$, gives a trade-off between update and query times.  The overall approach is a lazy rebuilding scheme using the state-of-the-art offline $k$-mismatch algorithm. In order to achieve a fast solution, we handle instances with many and few $2k$-mismatch occurrences differently. Basing on combinatorial insights developed in the sequence of papers on the offline and streaming versions of the $k$-mismatch problem~\cite{chan2020approximating,SODA:CFPST16,CKP,gawrychowski2018towards,GKKP}, we are able to achieve update time $\Ohtilde\Big(\frac{n}{k} +\sqrt{\frac{nk}{x}}\Big)$ and query time $\tO{x}$ for any trade-off parameter $x\in [1\dd k]$ provided at initialization.\footnote{Throughout this paper, we denote $[a\dd b]=\{i\in \mathbb{Z} : a\le i \le b\}$ and $[a\dd b)=\{i \in \mathbb{Z} : a \le i < b\}$.} To put the trade-off complexity in context, we note that, e.g., when $k = m$, this allows  achieving $U(n,k) = Q(n,k) = \tO{n^{2/3}}$, which improves upon an $\tO{n^{3/4}}$ bound presented in \cite{clifford2018upper} (where only the case of $k=m$ is considered).

We also show conditional lower bounds which are in most cases within subpolynomial factors of our upper bounds.  For the case where the text length is linear in the length of the pattern, we do this by reducing from the 3SUM conjecture~\cite{STOC:Patrascu10}.   However, in the case that the text is much longer than the pattern, our reduction requires the Online Matrix vector conjecture~\cite{STOC:HKNS15}. Interestingly the lower bound for the superlinear case is asymmetric between the query and update time.
%!TEX root = dynamic-kmismatch.tex
\section{Preliminaries}

In this section, we provide the required basic definitions. 
We begin with the string distance metric which will be used throughout.

\begin{definition}[Hamming Distance]
The Hamming distance between two strings $S$, $R$ of the same length is defined as $\HD(S,R) = \size{\{i:S[i]\neq R[i]\}}.$
\end{definition}

From this point forward, for simplicity of exposition, we will assume that the pattern is half the length of the text. All our upper bounds are straightforward to generalise to a text whose length is linear in the length of the pattern.
In Theorem~\ref{thm:longtext}, we show higher lower bounds for the case where the text is much longer than the pattern.

We can now define the central dynamic data structure problem we consider in this paper. 

\begin{definition}[Dynamic $k$-Mismatch Problem]
Let $P$ be a pattern of length $m$ and  $T$ be a text of length $n\le 2m$. 
For $i\in [0\dd n-m]$, a query $\Query(i)$ must return $\HD(T[i\ldots i+m),P)$ if $\HD(T[i\ldots i+m),P)\le k$, and $\infty$ otherwise. The queries can be
interspersed with updates of the form $\Update(S,i,x)$ which assign $S[i]:=x$,
where $S$ can be the pattern or the text.
\end{definition}

For the remainder of the paper, we use $Q(n,k)$ and $U(n,k)$ to be the time
complexity of $\Query$ and $\Update$, respectively.
If $n > 2m$, then a standard reduction yields $\Oh(Q(m,k))$-time queries,
$\Oh(U(m,k))$-time updates in $T$, and $\Oh(\frac{n}{m}Q(m,k))$-time updates in $P$.
%!TEX root = dynamic-kmismatch.tex
\newcommand{\X}{\mathcal{X}}

\section{Upper Bounds}
In this section, we provide three solutions of the dynamic $k$-mismatch problem. 
We start with a simple application of dynamic strings resulting in $\tO{k}$ query time and $\tO{1}$ update~time.

The data structure of Gawrychowski et al.~\cite{GKKLS} maintains a dynamic family $\X$ of strings of total length $N$ supporting the following updates:\footnote{This data structure is Las-Vegas randomized, and the running times are valid with high probability with respect to $N$. A deterministic version, using~\cite{MISC:AlsBroRau00} and deterministic dynamic dictionaries, has an $\Oh(\log N)$-factor overhead in the running times, which translates to an $\Oh(\log n)$-factor overhead in the query and update times of all our randomized algorithms for the dynamic $k$-mismatch~problem.}
\begin{itemize}
    \item Insert to $\X$ a given string $S$ (in time $\Oh(|S|+\log N)$).
    \item Insert to $\X$ the concatenation of two strings already in $\X$ (in time $\Oh(\log N)$).
    \item Insert to $\X$ an arbitrary prefix or suffix of a string already in $\X$ (in time $\Oh(\log N)$).
\end{itemize}
Queries include $\Oh(1)$-time computation of the longest common prefix of two strings in $\X$.

\begin{theorem}\label{thm:fastU}
There exists a Las-Vegas randomized algorithm for the dynamic $k$-mismatch problem satisfying $U(n,k)=\bigO{\log{n}}$
and $Q(n,k)=\bigO{k\log {n}}$ with high probability.
\end{theorem}
\begin{proof}
We maintain a dynamic string collection $\X$ of~\cite{GKKLS} containing $P$ and $T$.
Given that a string $S'$ resulting from setting $S[i]:=x$ in a string $S\in \X$
is the concatenation of a prefix $S[0\dd i)$, the new character $x$,
and a suffix $S[i+1\dd |S|)$, it is straightforward to construct $S'$
with $\Oh(1)$ auxiliary strings added to $\X$. 
Hence, we implement an update in $\Oh(\log n)$ time.

Armed with this tool, we perform dynamic $k$-mismatch queries by so-called ``kangaroo jumps''~\cite{TCS:LanVis86}. That, is we align the pattern with $T[i\dd i+m)$, where $i$ is the query position in the text $T$,
and we repeatedly extend the match we have found so far until we reach a fresh mismatch. Each longest common extension query can be implemented in $\Oh(\log n)$ time. For this, we extract the relevant suffixes of $P$ and $T$ (we insert them to $\X$ in $\Oh(\log n)$ time each) and ask for their longest common prefix (which costs $\Oh(1)$ time). As we stop once $k+1$ mismatches have been found or once we have reached the end of the text or pattern, the total query time is $\Oh(k \log n)$.
\end{proof}

\subsection{Faster Queries, Slower Updates}
Given the result above, a natural question is whether there exists an approach with an efficient
query algorithm, in return for a slower update algorithm. We answer affirmatively in this section
based on a recent work of Charalampopoulos et al.~\cite{charalampopoulos2020faster}.

\subsubsection{The PILLAR model}
\def\modelname{{\tt PILLAR}\xspace}
\def\fragmentco#1#2{[#1\dd#2)}
\def\accOpName{{\tt Access}\xspace}
\def\extractOpName{{\tt Extract}\xspace}
\def\lenOpName{{\tt Length}\xspace}
\def\perOpName{{\tt Period}\xspace}
\def\lceOpName{{\tt LCP}\xspace}
\def\lcbOpName{{\tt LCP$^R$}\xspace}
\def\lceOp#1#2{{\tt LCP}(#1, #2)}
\def\lcbOp#1#2{{\tt LCP}^R(#1, #2)}
\def\ipmOp#1#2{{\tt IPM}(#1, #2)}
\def\accOp#1#2{#1[#2]}

Charalampopoulos et al.~\cite{charalampopoulos2020faster} developed a generic static algorithm for the $k$-mismatch problem.
They formalized their solution using an abstract interface, called the \emph{\modelname model},
which captures certain primitive operations that can be
implemented efficiently in all settings considered in~\cite{charalampopoulos2020faster}.
Thus, we bound the
running times in terms of~\modelname operations---if the algorithm uses more time than
\modelname operations, we also specify the extra running time.

In the \modelname model, we are given a family of~strings $\X$ for preprocessing.
The elementary objects are fragments $X\fragmentco{\ell}{r}$ of~strings $X\in \X$.
Initially, the model provides access to each $X\in \X$ interpreted as $X\fragmentco{0}{|X|}$.
Other fragments can be obtained through an \extractOpName operation.
\begin{itemize}
    \item $\extractOpName(S,\ell,r)$: Given a fragment $S$ and positions $0 \le \ell \le r
        \le |S|$, extract the (sub)fragment $S\fragmentco{\ell}{r}$,
        which is defined as $X\fragmentco{{\ell'+\ell}}{\ell'+r}$ if $S=X\fragmentco{\ell'}{r'}$ for $X\in \X$.
\end{itemize}

\noindent
Furthermore, the following primitive operations are supported in the \modelname model:
\begin{itemize}
    \item $\lceOp{S}{T}$: Compute the length of~the longest common prefix of~$S$ and $T$.
    \item $\lcbOp{S}{T}$: Compute the length of~the longest common suffix of~$S$ and $T$.
    \item $\ipmOp{P}{T}$: Assuming that $|T|\le 2|P|$, compute the occurrences of $P$ in $T$, i.e., $\Occ(P,T)=\{i\in [0\dd |T|-|P|] : P = T[i\dd i+|P|)\}$ represented as an arithmetic progression.% with difference $\per(P)$).
    \item $\accOpName(S,i)$: Retrieve the character $\accOp{S}{i}$.
    \item $\lenOpName(S)$: Compute the length $|S|$ of~the string $S$.
\end{itemize}

Among several instantiations of the model, Charalampopoulos et al.~\cite[Section 7.3]{charalampopoulos2020faster} showed that the primitive \modelname operations can be implemented in $\Oh(\log^2 N)$ time on top of the data structure for dynamic strings~\cite{GKKLS}, which we recalled above. Consequently, we are able to maintain two dynamic strings $P$ and $T$ subject to character substitutions, achieving $\Oh(\log^2 n)$-time elementary \modelname operations and $\Oh(\log n)$-time updates.

\begin{corollary}\label{cor:pillar_impl}
Let $T$ be a dynamic string of length $n$ and $P$ be a dynamic string of length $m\le n$,
both of which can be updated via substitutions of single characters.
There exists a Las-Vegas randomized data structure supporting the \modelname operations on $\X = \{T,P\}$ in $\Oh(\log^2 n)$ time w.h.p.~and updates in  $\Oh(\log n)$ time w.h.p.
\end{corollary}

\subsubsection{The Static $k$-Mismatch Problem}
The (static) $k$-mismatch problem consists in computing $\Occ_k(P,T)=\{i\in [0\dd n-m] :\allowbreak \HD(P,T[i\dd i+m))\le k\}$, with each position $i\in \Occ_k(P,T)$ reported along with the corresponding Hamming distance $d_i := \HD(P,T[i\dd i+m))$.
Charalampopoulos et al.~\cite[Theorem 3.1 and Corollary 3.5]{charalampopoulos2020faster} proved that $\Occ_k(P,T)$ 
admits a compact representation: this set can be decomposed into $\Oh\big(\frac{n}{m}\cdot k^2\big)$ disjoint arithmetic progressions so that occurrences in a single progression share the same Hamming distance~$d_i$.
Moreover, all the non-trivial progressions (i.e., progressions with two or more terms) share the same difference.
The following algorithm gives this compact representation on the output.

\begin{theorem}[{\cite[Main Theorem 8]{charalampopoulos2020faster}}]\label{thm:pillar_k}
There exists a \modelname-model algorithm that, given a pattern $P$ of length $m$, a text $T$ of length $n\ge m$, and 
a positive integer $k\le m$, solves the $k$-mismatches problem in $\Oh(\frac{n}{m}\cdot k^2 \log \log k)$ time
using $\Oh(\frac{n}{m}\cdot k^2)$ \modelname operations.
\end{theorem}

\subsubsection{Warm-Up Algorithm}
Intuitively, the algorithm of \cref{thm:pillar_k} precomputes the answers to all queries $\Query(i)$ with $i\in [0\dd n-m]$.
Hence, a straightforward solution to the dynamic $k$-mismatch problem would be to maintain the data structure of \cref{cor:pillar_impl}, use the algorithm of \cref{thm:pillar_k} after each update, and then retrieve the precomputed answers for each query asked.
The data structure described below follows this strategy, making sure that the compact representation of $\Occ_k(P,T)$
is augmented with infrastructure for efficient random access. 
\begin{proposition}
There exists a Las-Vegas algorithm for the dynamic $k$-mismatch problem satisfying $U(n,k)=\bigO{k^2 \log^2 n}$
and $Q(n,k)=\bigO{\log\log n}$ with high probability.
\end{proposition}
\begin{proof}
We maintain a \modelname-model implementation of $\X=\{P,T\}$ using \cref{cor:pillar_impl};
this costs $\Oh(\log n)$ time per update and provides $\Oh(\log^2 n)$-time primitive \modelname operations.

Following each update, we use \cref{thm:pillar_k} so that a space-efficient representation
of $\Occ_k(P,T)$ is computed in $\Oh(k^2 \log^2 n)$ time (recall that $m = \Theta(n)$).
This output is then post-processed as described below. 
Let $q$ be the common difference of non-trivial arithmetic progression forming  $\Occ_k(P,T)$; we set $q=1$ if all progressions are trivial. 
Consider the indices $i\in [0\dd n-m]$ ordered by $(i\bmod q, i)$, that is, first by the remainder modulo $q$ and then by the index itself. 
In this ordering, each arithmetic progression contained in the output $\Occ_k(P,T)$
yields a contiguous block of indices $i$ with a common finite answer to queries $\Query(i)$.
The goal of post-processing is to store the sequence of answers using run-length encoding (with run boundaries
kept in a predecessor data structure). 
This way, for each of the $\Oh(k^2)$ arithmetic progressions in $\Occ_k(P,T)$,
the corresponding answers $\Query(i)$ can be set in $\Oh(\log \log n)$ time to the common value $d_i$ reported along with the progression.
In total, the post-processing time is therefore $\Oh(k^2 \log \log n)$.

At query time, any requested value $\Query(i)$ can be retrieved in $\Oh(\log \log n)$ time.
\end{proof}

\subsubsection{Structural Insight}
In order to improve the update time, we bring some of the combinatorial insight from~\cite{charalampopoulos2020faster}.

A string is \emph{primitive} if it is not a string power with an integer exponent strictly greater than~$1$.
For a non-empty string $Q$, we denote by $Q^\infty$ an infinite string obtained by concatenating infinitely many copies of $Q$.
For an arbitrary string $S$, we further set $\HD(S,Q^*)=\HD(S,Q^\infty[0\dd |S|))$. 
In other words, the $\HD(\cdot,\cdot^*)$ function generalizes $\HD(\cdot,\cdot)$ in that the second string is cyclically extended to match the length of the first one.
We use the same convention to define $M(S,Q^*)=\{i : S[i]\ne Q^\infty[i]\}=\{i : S[i] \ne Q[i\bmod |Q|]\}$.

\begin{proposition}[{\cite[Theorems 3.1 and 3.2]{charalampopoulos2020faster}}]\label{thm:char}
Let $P$ be a pattern of length $m$, let $T$ be a text of length $n\le \frac32 m$, and let $k\le m$ be a positive integer.
At least one of the following holds:
\begin{enumerate}
\item The number of $k$-mismatch occurrences of $P$ in $T$ is $|\Occ_k(P,T)| \le 864k$.
\item There is a primitive string $Q$ of length $|Q| \le \frac{m}{128k}$ such that $\HD(P,Q^*)<2k$.\label{it:per}
\end{enumerate}
Moreover, if $\Occ_k(P,T)\ne \emptyset$ and~\eqref{it:per}~holds, then a fragment $T'=T[\min \Occ_k(P,T)\dd m+\max \Occ_k(P,T))$
satisfies $\HD(T',Q^*)< 6k$ and every position in $\Occ_k(P,T')$ is a multiple~of~$|Q|$.
\end{proposition}

We also need a characterization of the values $\HD(P,T'[j|Q|\dd m+j|Q|))$.
\begin{proposition}[{\cite[Lemma 3.3 and Claim 3.4]{charalampopoulos2020faster}}]\label{prp:charmult}
Let $P$ be a pattern of length $m$, let $T$ be a text of length $n$,
and let $k\le m$ be a positive integer. 
For any non-empty string $Q$ and non-negative integer $j \le \frac{n-m}{|Q|}$,
we have \[\HD(P,T[j|Q|\dd m+j|Q|)) = |M(P,Q^*)|+|M(T,Q^*)\cap [j|Q|\dd m+j|Q|)| - \mu_j,\]
where \[\mu_j = \sum\limits_{\rho\in M(P,Q^*),\tau \in M(T,Q^*)\;:\; \tau = j|Q|+\rho} 2-\HD(T[\tau],P[\rho]).\]
\end{proposition}

\subsubsection{Improved Solution}
The idea behind achieving $\Oh(k\log^2 n)$ update time is to run \cref{thm:pillar_k}
once every $k$ updates, but with a doubled threshold $2k$ instead of $k$.
The motivation  behind this choice of parameters is that if
the current instance $P,T$ is obtained by up to $k$ substitutions from a past instance $\bar{P},\bar{T}$,
then $\HD(P,\bar{P})+\HD(T,\bar{T})\le k$ yields $\Occ_k(P,T) \sub \Occ_{2k}(\bar{P},\bar{T})$.
Consequently, the algorithm may safely return $\infty$ while answering $\Query(i)$
for any position $i\notin \Occ_{2k}(\bar{P},\bar{T})$.

If the application of \cref{thm:pillar_k} identifies few $2k$-mismatch occurrences, then 
we maintain the Hamming distances $d_i$ at these positions throughout the $k$ subsequent updates.
Otherwise, we identify $Q$ and $T'$, as defined in \cref{thm:char}, as well as the sets $M(P,Q^*)$, $M(T',Q^*)$,
and the values $\mu_j$ of \cref{prp:charmult}
so that the distances $\HD(P,T'[j|Q|\dd m+j|Q|))$ can be retrieved efficiently.

The latter task requires extending \cref{thm:pillar_k} so that the string $Q$
and the sets $M(P,Q^*)$, $M(T',Q^*)$ can be constructed
whenever there are many $k$-mismatch occurrences.
\begin{lemma}\label{lem:pillar_k}
    There exists a \modelname-model algorithm that, given a pattern $P$ of length $m$, a text $T$ of length $n\le \frac32 m$, and a positive integer $k\le m$, returns $\Occ_{k}(P,T)$ 
    along with the corresponding Hamming distances provided that $|\Occ_{k}(P,T)|\le 864k$,
    or, otherwise, returns the fragment $T'=T[\min \Occ_k(P,T)\dd m+\max \Occ_k(P,T))$,
    a string $Q$ such that $\HD(P,Q^*)<2k$, $\HD(T',Q^*)<6k$, and $\Occ_{k}(P,T')$ consists of multiples of $|Q|$,
    and sets $M(P,Q^*)$, $M(T',Q^*)$.
    The algorithm takes $\Oh(k^2 \log \log k)$ time plus $\Oh(k^2)$ \modelname operations.
\end{lemma}
\begin{proof}
First, we use \cref{thm:pillar_k} in order to construct $\Occ_k(P,T)$ in a compact representation
as $\Oh(k^2)$ arithmetic progressions. 
Based on this representation, both $|\Occ_k(P,T)|$ and $T'$ can be computed in $\Oh(k^2)$ time.
If $|\Occ_{k}(P,T)|\le 864k$, then $\Occ_k(P,T)$ is converted to a plain representation (with each position
reported explicitly along with the corresponding Hamming distance).
Otherwise, we use the $\texttt{Analyze}(P,k)$ procedure of~\cite[Lemma 4.4]{charalampopoulos2020faster}.
This procedure costs $\Oh(k)$ time in the \modelname model, and it detects
a structure within the pattern $P$ that can be of one of three types.
A possible outcome includes a primitive string $Q$ such that $|Q|\le \frac{m}{128k}$ and $\HD(P,Q^*)<8k$.
Moreover, the existence of a structure of either of the other two types contradicts $|\Occ_{k}(P,T)|\le 864k$
(due to~\cite[Lemmas 3.8 and 3.11]{charalampopoulos2020faster}),
and so does $2k \le \HD(P,Q^*)<8k$ (due to~\cite[Lemma 3.14]{charalampopoulos2020faster}).
Consequently, we are guaranteed to obtain a primitive string $Q$ such that  $|Q|\le \frac{m}{128k}$ and $\HD(P,Q^*)<2k$, which are precisely the conditions in the second case of \cref{thm:char}.
Thus, we conclude that $\HD(T',Q^*)<6k$ and that $\Occ_{k}(P,T')$ consists of multiples of $|Q|$.
It remains to report $M(P,Q^*)$ and $M(T',Q^*)$.
For this task, we employ~\cite[Corollary 4.2]{charalampopoulos2020faster},
whose time cost in the \modelname model is proportional to the output size, i.e., $\Oh(k)$ for both instances.
\end{proof}

We are now ready to describe the dynamic algorithm based on the intuition above.
Initially, we only improve the \emph{amortized} query time from $\Oh(k^2\log^2 n)$ to $\Oh(k\log^2 n)$.
\begin{proposition}\label{prp:amort}
There exists a Las-Vegas randomized algorithm for the dynamic $k$-mismatch problem satisfying $Q(n,k)=\bigO{\log \log n}$ and $U(n,k)=\bigO{k + \log n}$ with high probability, except that every $k$th update costs $\bigO{k^2 \log^2 n}$ time~w.h.p.
\end{proposition}
\begin{proof} 
The algorithm logically partitions its runtime into epochs, with $k$ updates in each epoch.
The first update in every epoch costs $\Oh(k^2 \log^2 n)$ time, and the remaining updates cost $\Oh(k + \log n)$ time.
A representation of $\X=\{P,T\}$ supporting the \modelname operations (\cref{cor:pillar_impl})
is maintained throughout the execution of the algorithm, while the remaining data is destroyed
after each epoch.

Once the arrival of an update marks the beginning of a new epoch,
we run the algorithm of \cref{lem:pillar_k} with a doubled threshold $2k$.
This procedure costs $\Oh(k^2 \log^2 n)$ time, and it may have one of two types of outcome.

The first possibility is that it returns a set $O:=\Occ_{2k}(P,T)$ of up to $1728k$ positions,
with the Hamming distance $d_i := \HD(P,T[i\dd i+m))$ reported along with each position $i\in O$.
Since $d_i > 2k$ for $i\notin O$ and any update may decrease $d_i$ by at most one,
we are guaranteed that $\Query(i)=\infty$ can be returned for $i\notin O$ for the duration of the epoch.
Consequently, the algorithm only maintains $d_i$ for $i\in O$. 
For each of the subsequent updates, the algorithm iterates over $i\in O$ and checks if $d_i$
needs to be changed: If the update involves $P[j]$, then both the old and the new value of $P[j]$
are compared against $T[i+j]$. Similarly, if the update involves $T[j]$ and $j\in [i\dd i+m)$, then both the old and the new value of $T[j]$ are compared against $P[i-j]$. 
Thus, the update time is $\Oh(k)$ and the query time is $\Oh(1)$.

The second possibility is that the algorithm of \cref{lem:pillar_k} results in a fragment $T'=T[\ell \dd r)$, a string $Q$, and the mismatching positions $M(P,Q^*)$ and $M(T',Q^*)$.
We are then guaranteed that each $2k$-mismatch occurrence of $P$ in $T$ starts at a position $i\in [\ell\dd r-m]$
congruent to $\ell$ modulo $|Q|$. We call these positions \emph{relevant}.
As in the previous case, $\Query(i)=\infty$ can be returned for irrelevant $i$ for the duration of the epoch. 
The Hamming distances $d_i$ at relevant positions are computed using \cref{prp:charmult}.
For this, we maintain $M(P,Q^*)$, $M(T',Q^*)$, and all non-zero values $\mu_j$ for $j\in [0\dd \lfloor\frac{r-\ell-m}{|Q|}\rfloor]$. Moreover, $M(T',Q^*)$ is stored in a predecessor data structure,
and each element of $M(T',Q^*)$  maintains its rank in this set.
Every subsequent update affects at most one element of $M(P,Q^*)$ or $M(T',Q^*)$,
so these sets can be updated in $\Oh(1)$ time. Maintaining the predecessor data structure costs further $\Oh(\log \log n)$ time, and maintaining the ranks costs up to $\Oh(\HD(T',Q^*))$ time.
In order to update the values $\mu_j$, we proceed as follows.
If the update involves a character $P[\rho]$, we iterate over $\tau \in M(T',Q^*)$. If $j={\tau-\rho}{|Q|}$ is an integer between $0$ and  $\frac{r-\ell-m}{|Q|}$,
we may need to update the entry $\mu_j$ (which costs constant time).
An update involving $T[\ell+\tau]$ is processed in a similar way.
Overall, the update time is $\Oh(\log n + \HD(T',Q^*) +\HD(P,Q^*))=\Oh(\log n + k)$
because $\HD(P,Q^*)+\HD(T',Q^*)<2k+6k+k=9k$ holds for the duration of the epoch.

As for the query $\Query(i)$, we return $\infty$ if $i$ is irrelevant, i.e., $i < \ell$, $i > r-m$, or $i \not\equiv \ell \pmod{|Q|}$.
Otherwise, we set $j=\frac{i-\ell}{|Q|}$ and, according to \cref{prp:charmult}, return
$|M(P,Q^*)| + |M(T',Q^*)\cap [j|Q|\dd j|Q|+m)|-\mu_j$.
The second term is determined in $\Oh(\log \log n)$ time using the predecessor data structure
on top of $M(T',Q^*)$ as well as the rank stored for each element of this set.
\end{proof}

Finally, we show how to achieve \emph{worst-case} $\Oh(k\log^2 n)$ update time.
\begin{theorem}\label{thm:fastQ}
There exists a Las-Vegas randomized algorithm for the dynamic $k$-mismatch problem satisfying $Q(n,k)=\bigO{\log \log n}$ and $U(n,k)=\bigO{k \log^2 n}$ with high probability.
\end{theorem}
\begin{proof}
We maintain two instances of the algorithm of \cref{prp:amort}, with updates forwarded to both instances,
but queries forwarded to a single instance that is currently \emph{active}.

The algorithm logically partitions its runtime into epochs, with $\frac12k$ updates in each epoch.
For the two instances, the time-consuming updates are chosen to be the first updates of every even and odd epoch,
respectively. Once an instance has to perform a time-consuming update, it becomes inactive (it buffers the subsequent updates and cannot be used for answering queries) and stays inactive for the duration of the epoch.
The work needed to perform the time-consuming update is spread across the time allowance for the first half of the epoch, with the time allowance for the second half of the epoch used in order to clear the accumulated backlog of updates (by processing updates at a doubled rate). 
During this epoch, the other (active) instance processes updates and queries as they arrive in $\Oh(k \log^2 n)$ 
and $\Oh(\log \log n)$ worst-case time, respectively.
\end{proof}

\subsection{Trade-off between Update Time and Query Time}
The next natural question is the existence of a trade-off between the run-times of~\cref{thm:fastQ,thm:fastU}.
Due to \Cref{thm:3sum_lower} (in \cref{sec:lb}), the answer is likely negative for $k\ll \sqrt{n}$.
Nevertheless, for $k\gg \sqrt{n}$, the trade-off presented below simultaneously achieves $Q(n,k),U(n,k)=k^{1-\Omega(1)}$.

We first recall some combinatorial properties originating from previous work on the $k$-mismatch problem~\cite{chan2020approximating,SODA:CFPST16,CKP,GKKP}.
The description below mostly follows~\cite[Section 3]{GKKP}.

\begin{definition}[\cite{SODA:CFPST16}]
Let $X$ be a string and let $d$ be a non-negative integer.
A positive integer $\rho\le |X|$ is a \emph{$d$-period} of $X$ if $\HD(X[\rho \dd |X|), X[0\dd |X| - \rho)) \le d$. 
\end{definition}

Recall that $\Occ_k(P,T)=\{i: \HD(P,T[i\dd i+m))\le k\}$ for a pattern $P$ and text $T$.
\begin{lemma}[\cite{SODA:CFPST16}]\label{lem:aperiodic}
If $i,i'\in \Occ_k(P,T)$ are distinct, then $\rho:=|i'-i|$ is a $2k$-period of $P$.
Moreover, if $n\le 2m$, then $\rho$ is a $(8k+\rho)$-period of $T[\min \Occ_k(P)\dd m+\max\Occ_k(P))$.
\end{lemma}

Recall that the $L_0$-norm of a function $f:\Z\to \Z$ defined as $\|f\|_0=|\{x : f(x)\ne 0\}|$.
The \emph{convolution} of two functions $f,g:\Z \to \Z$ with finite $L_0$-norms is a function $f*g: \Z \to \Z$ such that \[[f*g](i)=\sum_{j\in \Z} f(j)\cdot g(i-j).\]
For a string $X$ over $\Sigma$ and a symbol $c\in \Sigma$, the \emph{characteristic function} of $X$ and $c$
is $X_c : \Z \to \{0,1\}$ such that $X_c(i)=1$ if and only if $X[i]=c$.
For a string $X$, let $X^R$ denote $X$ reversed.
The \emph{cross-correlation} of strings $X$ and $Y$ over $\Sigma$ is a function $X\otimes Y:\mathbb{Z}\to \mathbb{Z}$
such that  \[X\otimes Y = \sum_{c\in \Sigma} X_c * Y_c^R.\]

\begin{fact}[{\cite[Fact 7.1]{CKP}}]\label{fct:cross}
For $i\in [m-1\dd n)$, we have $[T\otimes P](i)=|P|-\HD(P, T(i-m\dd i])$.
For $i<0$ and for $i\ge m+n$, we have $[T\otimes P](i)=0$.
\end{fact}
By \cref{fct:cross}, $[T\otimes P](i+m-1)$ suffices to compute $\HD(P,T[i\dd i+m))$ for $i\in [0\dd n-m]$.
The \emph{backward difference} of a function $f:\Z\to \Z$ due to $\rho\in \Z_+$ is $\Delta_\rho[f](i)=f(i)-f(i-p)$.

\begin{observation}[{\cite[Observation~7.2]{CKP}}]\label{obs:supp}
    If a string $X$ has a $d$-period $\rho$, then \[\sum_{c\in \Sigma} \|\Delta_\rho[X_c]\|_0 \le 2(d+\rho).\]
\end{observation}
Our computation of $T\otimes P$ is based on the following lemma:
\begin{lemma}[See {\cite[Lemma 6]{GKKP}}]\label{lem:second}
    For every pattern $P$, text $T$, and positive integer $\rho$, we have
    $\Delta_{\rho}[\Delta_{\rho}[T\otimes P]] = \sum_{c\in \Sigma} \Delta_{\rho}[T_c]*\Delta_{\rho}[P_c^R]$.
Consequently, for every $i\in \Z$,
\[[T\otimes P](i) = \sum_{j=0}^\infty (j+1)\cdot \left[\sum_{c\in \Sigma}\Delta_{\rho}[T_c]*\Delta_{\rho}[P_c^R]\right](i-j\rho).\]
\end{lemma}

\begin{theorem}
There exists a deterministic algorithm for the dynamic $k$-mismatch problem with $U(n,k)=\Oh\Big(\sqrt{\frac{nk}{x}} + \frac{n}{k}\Big)$
and $Q(n,k)=\tO{x}$, where $x$ is a trade-off parameter that can be set in $[1\dd k]$.
\end{theorem}
\begin{proof} 
We solve the problem using a lazy rebuilding scheme similar to that in the proof of \cref{thm:fastQ}.
Hence, we can afford update time $\Ohtilde\Big(n + k\sqrt{n}\Big)$ every $k$ updates.
Thus, if an incoming update marks the beginning of a new epoch (lasting for $k$ updates),
we run a (static) $2k$-mismatch algorithm~\cite{chan2020approximating,gawrychowski2018towards}, resulting in $O:=\Occ_{2k}(P,T)$ and the Hamming distances $d_i = \HD(P,T[i\dd i+m))$ for each $i\in O$. This takes $\tO{n + k\sqrt{n}}$ time.
As in the proof of \cref{prp:amort}, since $\Occ_k(P,Q)\sub O$ holds for the duration of the epoch,
we can safely return $\infty$ for $\Query(i)$ with $i\notin O$.
We distinguish two cases.
\begin{description}
\item[$|O|\le \frac{n}{k}$:] 
We maintain the distances $d_i$ for $i\in O$. As noted above, $\Occ_{k}(P,T)\sub O$ even after $k$ updates.
We now observe that any update requires only updating the mismatches for every element of $O$, with $\bigO{1}$ cost per element and $\bigO{\frac{n}{k}}$ total; the queries are handled by finding the answer stored for  $i\in O$, at $\tO{1}$ cost.
\item[$|O|> \frac{n}{k}$:] We set $\rho$ to be the distance between two closest elements of $O$;
we have $\rho\le k$  due to $|O|> \frac{n}{k}$. By \cref{lem:aperiodic}, $\rho$ is a $4k$-period of $P$ and a $17k$-period of $T':=T[\min O \dd m+\max O)$. Moreover, $\Occ_{k}(P,T)\sub O \sub [\min O\dd \max O]$
holds for the duration of the epoch, so all $k$-mismatch occurrences of $P$ in $T$ remain contained in $T'$.
\end{description}

We have thus reduced our problem to answering queries and performing updates for $P$ and $T'$. 
Moreover, we have a positive integer $\rho\le k$ which is initially a $4k$-period of $P$ and a $17k$-period of $T'$,
and, after $k$ updates, it remains a $6k$-period of $P$ and $19k$-period of $T'$.
Let us define the \emph{weight} of $c\in \Sigma$ as $\|\Delta_{\rho}[P^R_c]\|_0 + \|\Delta_{\rho}[T'_c]\|_0$;
by \cref{obs:supp}, the total weight across $c\in \Sigma$ remains $\Oh(k)$.

We proceed as follows. We maintain $\Delta_{\rho}(P^R_c) * \Delta_{\rho}(T'_c)$ for each letter $c\in \Sigma$ separately, and the sum $\Delta_{\rho}[\Delta_{\rho}[T\otimes P]]=\sum_{c \in \Sigma} \Delta_{\rho}(P^R_c) * \Delta_{\rho}(T'_c)$. 
For each remainder  $i\bmod \rho$, the values $\Delta_{\rho}[\Delta_{\rho}[T\otimes P]](i)$ are stored in a data structure that allows queries for prefix sums (both unweighted and weighted by $\floor{i/\rho}$)
so that $[T\otimes P](i)$ can be retrieved efficiently using \cref{lem:second}.
Every update to $P$ or $T$ incurs updates to $\Delta_{\rho}(P^R_c)$ or $\Delta_{\rho}(T'_c)$, in $\bigO{1}$ places in total (two for each letter involved in the substitution). 
We buffer the updates to those convolutions of (potentially) sparse functions during subepochs of $x$ updates, and then we recompute values of $\Delta_{\rho}(P^R_c) * \Delta_{\rho}(T'_c)$ amortized during the next $x$ updates. We fix a threshold value $t$ (specified later), and iterate through letters $c \in \Sigma$.
\begin{itemize}
\item If a letter $c$ had weight at least $t$ or accumulated at least $t$ updates, we recompute the corresponding convolution from scratch, at the cost of $\tO{n}$ time per each such \emph{heavy} letter.
\item Otherwise, updates are processed one by one, at the cost of $\tO{t}$ time per update.
\end{itemize}
There are $\bigO{\frac{k+x}{t}}$ heavy letters, which is $\tO{\frac{k}{t}}$ since $x\le k$. Thus, the total cost $\tO{\frac{nk}{t}  + x t}$ is minimized when $t = \sqrt{\frac{nk}{x}}$ and gives $\tO{\sqrt{nkx}}$ time per subepoch, or $\Ohtilde\Big(\sqrt{\frac{nk}{x}}\Big)$ time per update.

To perform queries, we retrieve $[T\otimes P](i+m-1)$ using \cref{lem:second} and the data structure maintaining 
$\Delta_{\rho}[\Delta_{\rho}[T\otimes P]]$ to recover the number of matches last time we stored the convolutions.
Next, we scan through the list of at most $2x$ updates to potentially update the answer.
\end{proof}

To put the trade-off complexity in context, we note that e.g., when $k = m$, it is possible to achieve $U(n,k), Q(n,k) = \tO{n^{2/3}}$. This improves over $\tO{n^{3/4}}$ presented in~\cite{clifford2018upper}.

%!TEX root = dynamic-kmismatch.tex
\section{Lower Bounds}\label{sec:lb}
In this section, we give conditional lower bounds for the dynamic
$k$-mismatch problem based on the 3SUM conjecture~\cite{STOC:Patrascu10}.
For the 3SUM problem, we use the following definition.

\newcommand{\G}{\mathbf{G}}
\begin{definition}[3SUM Problem]
Given three sets $A,B,C\sub [-N\dd N)$ of total size $\size{A}+\size{B}+\size{C}=n$,
decide whether there exist $a\in A$, $b\in B$, $c\in C$ such that $a+b+c=0$.
\end{definition}

Henceforth, we consider algorithms for the word RAM model with $w$-bit machine words,
where $w=\Omega(\log N)$.
In this model, there is a simple $\bigO{n^2}$-time solution for the 3SUM problem. This can be improved by log factors~\cite{Baran2008}, with the current record being $\bigO{(n^2/\log^{2} n)(\log \log n)^{\Oh(1)}}$ time~\cite{chan2019more}. 
\newcommand{\eps}{\varepsilon}
\begin{conjecture}[3SUM Conjecture]
For every constant $\eps > 0$, there is no Las-Vegas randomized algorithm
solving the 3SUM problem in $\Oh(n^{2-\eps})$ expected time.
\end{conjecture}

As a first step, we note that the 3SUM problem remains hard even if we allow for polynomial-time preprocessing 
of $A$. The following reduction is based on~\cite[Theorem 13]{LWWW16}.
\begin{lemma}\label{lem:prepr}
Suppose that, for some constants $d\ge 2$ and $\eps>0$, there exists an algorithm that, after preprocessing integers $n,N\in \Z_+$ and a set $A\sub [-N\dd N)$ in $\Oh(n^d)$ expected time, given sets $B,C\sub [-N\dd N)$ of total size $|A|+|B|+|C|\le n$, solves the underlying instance of the 3SUM problem in expected $\Oh(n^{2-\eps})$ time. Then, the 3SUM conjecture fails.
\end{lemma}
\begin{proof}
We shall demonstrate an algorithm solving the 3SUM problem in $\Oh(n^{2-\hat{\eps}})$ time, 
where $\hat{\eps} = \min(\frac12, \frac{\eps}{2(d-1)}) > 0$.
Let $g=\floor{n^{\frac{d-1.5}{d-1}}}$. We construct a decomposition $A=\bigcup_{i=1}^g A_i$ into disjoint subsets such that $|A_i|\le \ceil{\frac{1}{g}|A|}$ and $\max A_i < \min A_{i'}$ hold for $i,i'\in [1\dd g]$ with $i<i'$.  
Similarly, we also decompose $B=\bigcup_{j=1}^g B_j$ and $C=\bigcup_{k=1}^g C_k$.

Next, we construct $T = \{(i,j,k)\in [1\dd g]^3: \min A_i + \min B_j+\min C_k \le 0 \le  \max A_i + \max B_j+\max C_k\}$. Observe that  if $a+b+c=0$ for $(a,b,c)\in A\times B\times C$,
then the triple $(i,j,k)\in [1\dd g]^3$ satisfying $(a,b,c)\in A_i\times B_j\times C_k$ clearly belongs to $T$.
Moreover, $T$ can be constructed in $\Oh(g^2 \log g + |T|)$ time by performing a binary search
over $k\in [1\dd g]$ for all $(i,j)\in [1\dd g]^2$.
To provide a worst-case bound on this running time, we shall prove that $|T|=\Oh(g^2)$.
For this, let us define the \emph{domination} order $\prec$ on $[1\dd g]^3$ so that $(i,j,k)\prec  (i',j',k')$ 
if and only if $i<i'$, $j<j'$, and $k<k'$.
Observe that $T$ is an $\prec$-antichain and that $[1\dd g]^3$ can be covered with $\Oh(g^2)$ $\prec$-chains.
Hence, $|T|=\Oh(g^2)$ holds as claimed.

Let $\hat{n}:=\ceil{\frac{1}{g}|A|}+\ceil{\frac{1}{g}|B|}+\ceil{\frac{1}{g}|C|}=\Oh(\frac{n}{g})$.
We preprocess $(\hat{n},N,A_i)$ for each $i\in [1\dd g]$,
at the cost of $\Oh(g \hat{n}^d) = \Oh(\frac{n^d}{g^{d-1}})=\Oh(\frac{n^d}{n^{d-1.5}})=\Oh(n^{1.5})$ time.
Then, for each triple $(i,j,k)\in T$, we solve the underlying instance of the 3SUM problem,
at the cost of $\Oh(g^2 \hat{n}^{2-\eps})=\Oh(g^2 \frac{n^{2-\eps}}{g^{2-\eps}})
= \Oh(n^{2-\eps} g^{\eps})=\Oh(n^{2-\eps+\eps\frac{d-1.5}{d-1}})
= \Oh(n^{2-\frac{\eps}{2(d-1)}})$ expected time in total. As noted above, it suffices to return YES if and only if at least one of these calls returns YES.\@
\end{proof}

\newcommand{\TSP}{3SUM$^+$}
Our lower bounds also rely on the following variant of the 3SUM problem.
\begin{definition}[\TSP\ Problem]
	Given three sets $A,B,C\sub [-N\dd N)$ of total size $\size{A}+\size{B}+\size{C}=n$,
	report all $c\in C$ such that $a+b+c=0$ for some $a\in A$ and $b\in B$.
\end{definition}

The benefit of using \TSP\ is that it remains hard for $N\ge n^{2+\Omega(1)}$ (as shown in~\cite{HU17});
in comparison, regular 3SUM is known to be hard only for $N\ge n^3$.
The following proposition generalizes the results of~\cite{HU17} (allowing for preprocessing of $A$); its proof relies on the techniques of~\cite{Baran2008}.

\begin{proposition}\label{thm:hsu}
Suppose that, for some constants $d\ge 2$ and $\eps,\delta>0$, there exists an algorithm
that, after $\Oh(n^d)$-time preprocessing of integers $n,N\in \Z_+$, with $N\le n^{2+\delta}$,
and a set $A\sub [-N\dd N)$, given sets $B,C\sub [-N\dd N)$ of total size $|A|+|B|+|C|\le n$,
solves the underlying \TSP\ instance in expected $\Oh(n^{2-\eps})$ time.
Then, the 3SUM conjecture fails.
\end{proposition}
\begin{proof}
We shall demonstrate an algorithm violating the 3SUM conjecture via \cref{lem:prepr}.
If the input instance already satisfies $N\le n^{2+\delta}$, there is nothing to do.
Thus, we henceforth assume $N>n^{2+\delta}$.
Let $v=\ceil{\log 3N}$ and $u=\floor{\log n^{2+\delta}}$.
In the preprocessing, we draw a uniformly random \emph{odd}
integer $\alpha\in [0\dd 2^{v})$,
which defines a hash function $h: \Z \to [0\dd 2^u)$
with $h(x)=\floor{\frac{\alpha x\bmod 2^{v}}{2^{v-u}}}$ for $x\in \Z$.
The key property of this function is that $(h(a)+h(b)+h(c)-h(a+b+c))\bmod 2^u\in \{0,-1,-2\}$ holds for all $a,b,c\in \Z$. 
At the preprocessing stage, we also preprocess $(n,2^u,h(A))$ for the hypothetical \TSP\ algorithm
(note that $2^u \le n^{2+\delta}$).
Overall, the preprocessing stage costs $\Oh(n^d)$ time.

In the main phase, we solve the following \TSP\ instances, each of size at most $n$
and over universe $[-2^u\dd 2^u)$, denoting $X+y:=\{x+y: x\in X\}$:
\begin{itemize}
	\item $(h(A),h(B),h(C))$,
	\item $(h(A),h(B),h(C)-2^u+2)$, 
	\item $(h(A),h(B),h(C)-2^u+1)$, 
	\item $(h(A),h(B),h(C)-2^u)$, 
	\item $(h(A),h(B)-2^u,h(C)-2^u+2)$,
	\item $(h(A),h(B)-2^u,h(C)-2^u+1)$, 
	\item $(h(A),h(B)-2^u,h(C)-2^u)$;
\end{itemize}
this step costs $\Oh(n^{2-\eps})$ time.
Combining the results of these calls,
in $\Oh(n)$ time we derive \[S:=\{c\in C : h(a)+h(b)+h(c)\in \{0,2^u-2, 2^u-1,2^u,2\cdot 2^u-2,2\cdot 2^u-1,
2\cdot 2^u\}\}.\]
Finally, for each $c\in S$, we check in $\Oh(n)$ time whether $a+b+c=0$ holds for some $a\in A$ and $b\in B$. Upon encountering the first witness $c\in S$, we return YES.\@
If no witness is found, we return NO.\@

Let us analyze the correctness of this reduction.
If we return YES, then clearly $a+b+c=0$ holds for some $a\in A$, $b\in B$, and $c\in C$.
For the converse implication, suppose that $a+b+c=0$ holds for some $a\in A$, $b\in B$, and $c\in C$.
Then, $(h(a)+h(b)+h(c)-h(a+b+c))\bmod 2^u=(h(a)+h(b)+h(c))\bmod 2^u \in \{0,-1,-2\}$.
Given that $h(a),h(b),h(c)\in [0\dd 2^u)$, 
this means that $h(a)+h(b)+h(c)\in \{0,2^u-2, 2^u-1,2^u,2\cdot 2^u-2,2\cdot 2^u-1,
2\cdot 2^u\}$, i.e., $c\in S$. Consequently, we are guaranteed to return 
YES while processing $c\in S$ at the latest.

It remains to bound the expected running time.
For this, it suffices to prove that there are, in expectation, $\Oh(n^{1-\delta})$ triples $(a,b,c)\in A\times B\times C$ such that $a+b+c\ne 0$ yet $h(a)+h(b)+h(c) \in \{0,2^u-2, 2^u-1,2^u,2\cdot 2^u-2,2\cdot 2^u-1,
2\cdot 2^u\}$ (in particular, this means that, in expectation, $S$ contains at most $\Oh(n^{1-\delta})$ non-witnesses; verifying all of them costs $\Oh(n^{2-\delta})$ expected time in total).
Specifically, we shall prove that each triple satisfies the aforementioned condition with probability $\Oh(n^{-2-\delta})$.

Due to the fact that $(h(a)+h(b)+h(c)-h(a+b+c))\bmod 2^u \in \{0,-1,-2\}$,
the bad event holds only if $a+b+c\ne 0$ yet $h(a+b+c)\in \{0,1,2,2^u-2,2^u-1\}$.
Let $a+b+c=2^t \beta$ for an integer $t\in \Z_{\ge 0}$ and odd integer $\beta\in \Z$.
Due to $|a+b+c|\le 3N \le 2^v$, we must have $t\in [0\dd v)$.
\begin{itemize}
\item If $t > v-u+1$, then $h(a+b+c)$ is uniformly random odd multiple of $2^{t-v+u}$ within $[0\dd 2^u)$. 
Hence, $\Pr[h(a+b+c)\in\{0,1,2,2^u-2,2^u-1\}]=0$.
\item If $t= v-u+1$, then $h(a+b+c)$ is a uniformly random odd multiple of $2$  within $[0\dd 2^u)$.
Hence, $\Pr[h(a+b+c)\in \{0,1,2,2^u-2,2^u-1\}]\le \frac{2}{2^{v-2}}=\frac{8}{2^v}$.
\item If $t= v-u$, then $h(a+b+c)$ is a uniformly random odd multiple of $1$ within $[0\dd 2^u)$.
Hence, $\Pr[h(a+b+c)\in\{0,1,2,2^u-2,2^u-1\}]\le \frac{2}{2^{v-1}}=\frac{4}{2^v}$.
\item If $t < v-u$, then $h(a+b+c)$ is a uniformly random element of  $[0\dd 2^u)$.
Hence, $\Pr[h(a+b+c)\in \{0,1,2,2^u-2,2^u-1\}]\le \frac{5}{2^{v}}$.
\end{itemize}
Overall, the probability is bounded by $\frac{8}{2^v} = \Oh(n^{-2-\delta})$.
\end{proof}

We are now in a position to give the lower bound for the dynamic $k$-mismatch problem.

\begin{theorem}\label{thm:3sum_lower}
	Suppose that, for some constants $p>0$, $\eps>0$, and $0<c<\frac12$,
	there exists a dynamic $k$-mismatch algorithm that solves instances satisfying $k=\ceil{m^c}$
	using initialization in $\Oh(n^p)$ expected time, updates in $\bigO{k^{1-\eps}}$ expected time, and	queries in $\bigO{k^{1-\eps}}$ expected time.
	Then, the 3SUM conjecture fails.
	This statement remains true when updates are allowed in either the pattern or the text (but not both).
\end{theorem}
\begin{proof}
	We shall provide an algorithm contradicting \cref{thm:hsu} for $\delta=\frac{1-2c}{2c}$
	and $d=\frac{p}{c}$.
	Suppose that the task is to solve a size-$\hat{n}$ instance of the \TSP\ problem
	with $A,B,C\sub [-N\dd N)$.
	We set $m=\ceil{\hat{n}^{1/c}}$ (so that $k=\ceil{m}^c \ge \hat{n}$), and $n=2m$,
	and we initialize a pattern to $P=\texttt{0}^m$ and a text to $T=\texttt{0}^{n}$.
	Observe that $m \ge \hat{n}^{1/c} \ge N^{\frac1{c(2+\delta)}} = N^{\frac{2}{1+2c}}$.
	If $N^{\frac{2}{1+2c}} < 2N$, then $N=\Oh(1)$, and we can afford to solve the \TSP\ instance naively.
	Otherwise, we are guaranteed that $m\ge 2N$, 
	and we proceed as follows:
    \begin{itemize}
		\item we set $P[a+N]:=\texttt{1}$ for each $a\in A$;
		\item we set $T[2N-b]:=\texttt{1}$ for each $b\in B$.
	\end{itemize}
	Finally, for each element $c\in C$, we perform a query at position $c+N$,
	and report $c$ if and only if $\HD(P,T[c+N\dd c+N+m))< \HD(P, \texttt{0}^m)+\HD(T[c+N\dd c+N+m),\texttt{0}^m)$. Due to the fact that $\HD(P, \texttt{0}^m)+\HD(T[c+N\dd c+N+m),\texttt{0}^m) \le |A|+|B|\le \hat{n}=k$,
	this can be decided based on the answer to the query.
	Equivalently, we report $c\in C$ if and only if $P[i]=T[c+N+i]=\texttt{1}$ holds for some $i\in [0\dd N)$,
	i.e., $T[a+c+2N]=1$ for some $a\in A$, or, equivalently, $-a-c\in B$, i.e., $a+b+c=0$ for some $b\in B$.
	This proves the correctness of the algorithm.

	As for the running time, note that the preprocessing phase costs 
	$\Oh(n^p)=\Oh(m^p)=\Oh(\hat{n}^{\frac{p}{c}})=\Oh(\hat{n}^d)$ expected time.
	The main phase, on the other hand, involves $\Oh(\hat{n})$ updates and queries,
	which cost $\Oh(\hat{n}\cdot k^{1-\eps})=\Oh(\hat{n}^{2-\eps})$ expected time in total.
	By \cref{thm:hsu}, this algorithm for \TSP\ would violate the 3SUM conjecture.

	If the updates are allowed in the text only, we set up the pattern during the preprocessing phase
	based on the fact that the target value of $P$ depends on $A$ only.
	If the updates are allowed in the pattern only, we exchange the roles of $A$ and $B$
	and set up the text during the preprocessing phase.
\end{proof}

Next, we note that the lower bound can be naturally extended to $c\ge \frac12$.

\begin{corollary}\label{cor:joint}
	Suppose that, for some constants $p>0$, $\eps>0$, and $0<c\le 1$,
	there exists a dynamic $k$-mismatch algorithm that solves instances satisfying $k=\ceil{m^c}$
	using initialization in $\Oh(n^p)$ expected time, updates in $\bigO{\min(\sqrt{m},k)^{1-\eps}}$ expected time, and	queries in $\bigO{\min(\sqrt{m},k)^{1-\eps}}$ expected time.
	Then, the 3SUM conjecture fails.
	This statement remains true when updates are allowed in either the pattern or the text (but not both).
\end{corollary}
\begin{proof}
	When $c<\frac12$, the result holds directly due to \cref{thm:3sum_lower}.
	When $c\ge \frac12$, we prove that the 3SUM conjecture would be violated
	through \cref{thm:3sum_lower} with $\hat{c}=\frac{1-\eps}{2-\eps}$ and $\hat{\eps}=\frac{\eps}{2}$.
	Since the $\hat{k}$-mismatch problem with $\hat{k}=\ceil{m}^{\hat{c}}$ can be simulated using an instance of the $k$-mismatch problem with $k=\ceil{m}^c$,
	we note that, in the former setting, the queries and updates can be hypothetically implemented in $\Oh((\sqrt{m})^{1-\eps})
	= \Oh(\hat{k}^{\frac{1-\eps}{2\hat{c}}})=\Oh(\hat{k}^{\frac{2-\eps}{2}})=\Oh(\hat{k}^{1-\hat{\eps}})$ expected time, violating the 3SUM conjecture via \cref{thm:3sum_lower}.
\end{proof}

\subsection{Lower Bound for \boldmath $m\ll n$}
While most of the work in this paper focuses on the case where the length of
the pattern is linear in the length of the text, for completeness, we provide
a lower bound that is only of interest when the pattern is considerably shorter.
Our lower bound is
conditioned on the Online Matrix-Vector Multiplication
conjecture~\cite{STOC:HKNS15}, which is often used
in the context of dynamic algorithms. 

In the Online Boolean Matrix-Vector Multiplication (OMv) problem, we
are given as input a Boolean matrix $M\in\{0,1\}^{n\times n}$. Then, a
sequence $n$ vectors $v_1, \ldots, v_n\in \{0,1\}^n$ arrives in
an online fashion. For each such vector $v_i$, we are required to
output $Mv_i$ before receiving $v_{i+1}$.

\begin{conjecture}[OMv Conjecture~\cite{STOC:HKNS15}]\label{conj:OMv}
  For any constant $\epsilon>0$, there is no
  $\Oh(n^{3-\epsilon})$-time algorithm that solves OMv correctly with
  probability at least $\frac23$.
\end{conjecture}

We use the following simplified version of~\cite[Theorem
2.2]{STOC:HKNS15}.

\begin{theorem}[\cite{STOC:HKNS15}]\label{thm:gOMv}
	Suppose that, for some constants $\gamma,\eps>0$,
	there is an algorithm that, given as input a matrix 
  $M\in \{0,1\}^{p\times q}$, with $q=\floor{p^{\gamma}}$, preprocesses $M$ in
  time polynomial in $p\cdot q$, and then, presented with a vector $v\in \{0,1\}^q$,
  computes $Mv$ in time $\Oh(p^{1+\gamma-\eps})$
  correctly with probability at least $\frac23$.
  Then, the OMv conjecture fails.
\end{theorem}

\cref{thm:gOMv} lets us derive our lower bound for the dynamic $k$-mismatch problem.

\begin{theorem}\label{thm:longtext}
	Suppose that, for some constants $\gamma,\eps>0$,
	there is a dynamic $k$-mismatch algorithm that solves instances
	satisfying $k= 2\floor{(\frac{n}{m})^\gamma}$,
	with preprocessing in $\Oh(n^{\Oh(1)})$ time, updates of the pattern in $\bigO{(\frac{n}{m})^{1-\eps}}$ time,
	and queries in $\bigO{k^{1-\eps}}$ time, providing correct answers with high probability.
	Then, the OMv conjecture fails.
\end{theorem}
\begin{proof}
    Given a matrix $M\in \{0,1\}^{p\times q}$, we 
	set $m = 3q$, $n=3pq$, and $k=2q$ (so that $k=2q=2\floor{p^\gamma}=2\floor{(\frac{n}{m})^\gamma}$ holds).
	At the preprocessing phase, we initially set $P=\texttt{0}^{m}$ and $T=\texttt{0}^{n}$.
	As for the text, for each $i\in [0\dd p)$ and $j\in [0\dd q)$,
	we set $T[3iq+3j\dd 3iq+3j+3)$ to $\texttt{100}$ if $M[i,j]=0$
	and to $\texttt{111}$ if $M[i,j]=1$.

	When a vector $v$ arrives, the pattern is set as follows:
	for each $j\in [0\dd q)$, we set $P[3j\dd 3j+3)$ to $\texttt{001}$ if $v[j]=0$
	and  to $\texttt{111}$ if $v[j]=1$.
	This requires $\bigO{q}$ update calls to our dynamic data structure. Queries are then made
	at position $im$ for all $i\in [0\dd p)$. By definition of the Hamming distance,
	$\HD(\texttt{100},\texttt{001})=\HD(\texttt{100},\texttt{111})=\HD(\texttt{111},\texttt{001})=2$, whereas $\HD(\texttt{111},\texttt{111})=0$; therefore, the only time
	that the returned Hamming distance will be
	less than $k=2q$ is when $M[i,j]=v[j]=1$ for some $j\in [1\dd q]$, i.e., $(Mv)[i]=1$.
	Therefore, the OMv product can
	be computed by making $\bigO{p}$ queries and $\bigO{q}$ updates to the dynamic $k$-mismatch data structure as before.
	The total cost of these operations is $\Oh(pq^{1-\eps}+qp^{1-\eps})
	= \Oh(p^{1+\gamma(1-\eps)}+p^{1+\gamma-\eps})=\Oh(p^{1+\gamma-\min(1,\gamma)\eps})$.
	By \cref{thm:gOMv} with $\hat{\eps}=\min(1,\gamma)\eps$, this would violate the OMv conjecture.
\end{proof}

\section*{Acknowledgements}
We are grateful to Ely Porat and Shay Golan for insightful conversations about the dynamic $k$-mismatch problem at an early stage of this work.
\bibliographystyle{plainurl}
\bibliography{bibliography}
\end{document}